\documentclass[11pt, reqno]{article}
\usepackage[letterpaper, total={6.5in, 9in}, 
]{geometry}

\usepackage{amsmath, amsthm, amscd, amsfonts, thmtools, amssymb, graphicx, xcolor, color}
\usepackage[bookmarksnumbered, colorlinks, plainpages=false]{hyperref}
\hypersetup{colorlinks=true,citecolor=magenta,linkcolor=blue} 

\allowdisplaybreaks
\usepackage[T1]{fontenc}
\usepackage{lmodern}
\usepackage{newtxmath,newtxtext}
\usepackage[page]{appendix}
\usepackage{array}
\usepackage{colortbl}
\newcolumntype{C}{>{$}c<{$}}
\newcolumntype{G}{>{\columncolor[gray]{0.8}$}c<{$}}
\usepackage{multirow}

\usepackage[
backend=biber,
style=alphabetic,
sorting=ynt
]{biblatex}
\usepackage{algorithm}
\usepackage{algpseudocode}
\floatname{algorithm}{Algorithm}

\newcommand{\matmul}{\textsc{MatMul}}

\newcommand{\bin}{\mathsf{bin}}

\newcommand{\ba}{\boldsymbol{\alpha}}
\newcommand{\bb}{\boldsymbol{\beta}}
\newcommand{\bg}{\boldsymbol{\gamma}}

\newcommand{\bA}{\mathbf{A}}

\newcommand{\bB}{\mathbf{B}}
\newcommand{\bU}{\mathbf{U}}
\newcommand{\bV}{\mathbf{V}}
\newcommand{\bW}{\mathbf{W}}

\newcommand{\bC}{\mathbf{C}}

\newcommand{\cD}{\mathcal{D}}
\newcommand{\cE}{\mathcal{E}}

\newcommand{\cI}{\mathcal{I}}

\newcommand{\cN}{\mathcal{N}}

\newcommand{\cR}{\mathcal{R}}

\newcommand{\R}{\mathbb{R}}
\newcommand{\C}{\mathbb{C}}

\newcommand{\Z}{\mathbb{Z}}
\newcommand{\E}{\mathbb{E}}

\newcommand{\norm}[2]{\left\lVert #1 \right\rVert_{#2}}
\newcommand{\abs}[1]{\left| #1 \right|}
\newcommand{\para}[1]{\left( #1 \right)}
\newcommand{\set}[1]{\left\{#1\right\}}
\newcommand{\ip}[1]{\left\langle{ #1 }\right\rangle}

\newcommand{\supp}{\mathrm{supp}}

\newcommand{\error}[2]{\mathrm{Collisions}_{G}^{#2}\left( #1 \right)}
\newcommand{\buck}[2]{\mathrm{Bucket}_{G}^{#2}\left( #1 \right)}

\newcommand{\ex}[2]{\underset{#1}{\mathbb{E}}\left[ #2 \right]}

\newcommand{\ti}[1]{\widetilde{#1}}
\newcommand{\eps}{\varepsilon}

\newtheorem{theorem}{Theorem}[section]
\newtheorem*{theorem*}{Theorem}
\newtheorem{lemma}[theorem]{Lemma}

\newtheorem{corollary}[theorem]{Corollary}
\theoremstyle{definition}
\newtheorem{definition}[theorem]{Definition}

\newtheorem{conjecture}[theorem]{Conjecture}

\theoremstyle{remark}
\newtheorem{remark}[theorem]{Remark}
\numberwithin{equation}{section}

\newif\ifnotanonymous
\notanonymoustrue

\begin{document}
\setcounter{page}{1}

\title{
(Approximate) Matrix Multiplication via Convolutions}

\ifnotanonymous
    \author{
    Yahel Uffenheimer \thanks{Hebrew University of Jerusalem. Supported by ERC Starting Grant (CODY 101039914).} 
    \and 
    Omri Weinstein \thanks{Hebrew University of Jerusalem. Supported by ISF grant \#3011005535 and ERC Starting Grant (CODY 101039914).}}
\else
    \author{Anonymous Authors}
\fi

\maketitle

\begin{abstract}
We study the capability of the Fast Fourier Transform (FFT) to accelerate exact and approximate matrix multiplication \emph{without} using Strassen-like divide-and-conquer. We present a simple exact algorithm running in $O(n^{2.89})$ time, which only sums a few convolutions (FFTs) in $\mathbb{Z}_m^k$, building on the work of Cohn, Kleinberg, Szegedy and Umans \cite{CKSU05}. As a corollary, combining this algorithm with linear sketching breaks the longstanding linear speed–accuracy tradeoff for ``combinatorial'' approximate matrix multiplication (AMM, \cite{pagh13,sar06,cw13}), achieving error $\tfrac{1}{r^{1.1}} \|\mathbf{A}\|_F^2 \|\mathbf{B}\|_F^2$ in $O(rn^2)$ time, using nothing but FFTs.

Motivated by the rich literature for approximating polynomials, our main contribution in this paper is extending the group-theoretic framework of Cohen and Umans (2003) to \emph{approximate} matrix multiplication (AMM). Specifically, we introduce and study an approximate notion of the \emph{Triple Product Property}, which in the abelian case is equivalent to  finding a Sumset which minimizes (multi-)intersections with an arithmetic progression. We prove tight bounds on this quantity for abelian groups (yielding a simple and practical AMM algorithm via polynomial multiplication), and establish a weaker lower bound for non-abelian groups, extending a lemma of Gowers. 
Finally, we propose a concrete approach that uses \emph{low-degree approximation} of multivariate polynomials for AMM, which we believe will lead to practical, non-asymptotic AMM algorithms in real-world applications, most notably LLM inference.

\end{abstract}

\section{Introduction}

Matrix multiplication (\matmul) is a fundamental operation across all fields of science and technology, and underlies most industry-scale applications. Above all,  
 \matmul s are the computation and performance bottleneck of training and inference of deep neural networks, where both forward and backpropagation rely on giant matrix multiplications --- for example, multiplying $16K\times 16K$ matrices is now considered a prerequisite in any LLM \cite{Li24LLMs,overview_of_llms}.  As such, the continual increase in computations and energy facilitating AI breakthroughs poses a real problem of scalability for the field.

The discovery of \emph{fast matrix multiplication} (FMM) --asserting that the product $\bA\bB$ of two $n\times n$ real matrices can be 
computed in $n^\omega \sim O(n^{2.37})\ll n^3$ time \cite{s86,w12,l14, duan2023faster}---had a profound impact on theoretical computer science and algorithm design. 
Alas, FMM algorithms are unlikely to be practical on any imaginable hardware, as they \emph{inherently rely on recursive divide-and-conquer} schemes, which create memory and IO-bottlenecks \cite{KAA20} and lead to highly asymptotic runtimes, earning them the infamous name ``galactic algorithms" \cite{LR10}. Another drawback of FMM algorithms in the context of data-driven applications, is that they are \emph{agnostic} to the input matrices -- recursion hinders the algorithm's ability to exploit structure in the input matrices (e.g., sparsity or invariances), except in a black-box fashion \cite{YusterZwick2005}.

In attempt to capture practical algorithms, a concurrent line of work has focused on ``combinatorial" MatMul algorithms \cite{FourRussians70, Wil07, Lingas09,  Chan15, yu18, AFKLM23}, which \emph{avoid} Strassen-like divide-and-conquer and algebraic decompositions, and have non-asymptotic runtimes \cite{BDHS15, KS17}. 
While the term ``combinatorial" is open to interpretation (especially when working over $\R$), we understand it as algorithms avoiding divide and conquer. An open question in the field is whether combinatorial algorithms can achieve truly subcubic runtime:

\begin{conjecture}[\cite{AFKLM23, satta94, lee01}]\label{conj_comb_MM}
    There is no $n^{3-\Omega(1)}$-time combinatorial algorithm for matrix multiplication, even for Boolean matrices.
\end{conjecture}

Bypassing \autoref{conj_comb_MM} in pursuit of more practical algorithms, has initiated a long line of research on \emph{approximate matrix multiplication} (AMM), which studies the best \emph{speed-accuracy} tradeoff achievable by ``combinatorial'' algorithms in close to \emph{quadratic} time. More precisely, for a prescribed parameter $r < n$, the  goal is to produce, in $O(rn^2)$ time, a matrix $\bC \in \R^{n\times n}$, which $\eps$-approximates $\bA\bB$ in the Frobenius norm, where the (normalized) error $\eps \rightarrow 0$ as $r\rightarrow n$. Essentially all known AMM algorithms use randomized \emph{sketching or sampling}  techniques \cite{sar06, drineas06, pagh13, magen10, cw13, CNW16, CL99}, and the state-of-art after more than 20 years of research is a \emph{linear} speed-accuracy tradeoff: 
\begin{align}\label{eq_AMM_sota}
    \|\bC - \bA\bB\|^2_F 
    \leq  
    O\left( \min\left\{ \frac{1}{r}\norm{\bA}{F}^2\cdot \norm{\bB}{F}^2 \; , \; \frac{n}{r}\|\bA\bB\|^2_F \right\} \right),  
    \tag{AMM error}
\end{align}
where $\bC$ must be produced in $O(rn^2)$ (randomized) time.  We usually refer to the \textit{normalized} error, given by $\frac{\norm{\bC-\bA\bB}{F}^2}{\norm{\bA}{F}^2\cdot \norm{\bB}{F}^2}$, where the above bound becomes $1/r$ (we call this the ``linear bound''). The first error term can be obtained via the standard ``sketch-and-solve'' algorithm (e.g., CountSketch \cite{cw13}), while the second one is obtained by a clever \emph{output-sensitive} variation of CountSketch, using FFT (TensorSketch \cite{pagh13}). The two bounds in \eqref{eq_AMM_sota} 
are generally incomparable, but they coincide for the (hardest known) distribution of random Gaussian matrices.  A conceptual limitation of all aforementioned AMM algorithms (except \cite{CL99} which applies only to nonnegative matrices), is that they use \emph{compression} techniques, i.e., produce a low-rank or sparse output matrix $\bC$. In data-driven applications, most notably LLM training and inference, this is a severe limitation since compression crucially decreases the number of trainable parameters (see \cite{STL25} and references therein).
Moreover, a recent result of \cite{OP25} proves that for compression-based algorithms, in a setup where Alice and Bob can send $O(rn)$ bits representing their respective input matrices $\bA,\bB \in \R^{n\times n}$ to a ``referee'' who must then compute the output based on their messages, the error bound in \eqref{eq_AMM_sota} is tight for  random Gaussian (or $\pm 1$) matrices.

\subsection{Our Contributions}

We view this work as an extension of the Group Theoretic approach developed by \cite{CU03}, originally for FMM, to the approximate setting. Put simply, we study AMM algorithms using (multi-variate) polynomial multiplication (and nothing else), avoiding Strassen style divide-and-conquer\footnote{Divide and conquer in bilinear algorithms essentially boils down to the computation of a Kronecker power matrix with a vector, $\bU^{\otimes k}v$, which can be done in $\Theta(mn^k)$ assuming $\bU$ is a $n\times m$ sized matrix with $m<n$, or $\Theta(kn^{k+1})$ when $m=n$. Suppose we treat $n$ and $m$ as constants, which leads the runtime to be asymptotically equivalent to the that of FFT (or even faster, when $m<n$, since $m$ does not scale with $k$). The key difference with FFT remains --- FFT has widely available and highly optimized implementations with very small constants, and has \textbf{non-asymptotic} runtime. Thus, for real-life matrix sizes, it is incomparable in efficiency to generic bilinear algorithms.}.

Our first contribution is to put in writing a simple observation that allows one to obtain FFT-based algorithms from the group theoretic approach. Applying this observation on a construction proposed by \cite{CKSU05} leads to the following result:

\begin{theorem}[Faster ``Combinatorial'' MM]\label{thm_main_exactMM_informal}
    There is an exact, deterministic algorithm for computing the product of two $n\times n$ matrices $\bA\bB$, in $O(n^{2.89})$ time. The algorithm only aggregates several convolutions in $\Z_m^k$ which are computed via FFT. 
\end{theorem}

Unfortunately, despite FFT being practical, the algorithm is still highly asymptotic. This is our main motivation to explore approximate versions, where the hope is to exploit some structure of the underlying construction. As a proof-of-concept, we show that a black-box combination of the algorithm in \autoref{thm_main_exactMM_informal} with standard linear sketching methods leads to the first super-linear speed-accuracy tradeoff for AMM:

\begin{theorem}[Faster AMM via Sketching Convolutions, Informal version of \autoref{thm:sketch-n-solve}]\label{thm_S&S_AMM}
For any $r\le n$, there exists a randomized algorithm that, given two square $n\times n$ matrices $\bA,\bB$, produces in time $\tilde{O}(n^2 r)$ a matrix $\bC$, computed by using FFTs, such that  w.h.p,  $\|\bC - \bA\bB\|_F^2 = O(\frac{1}{r^{1+\delta}}\norm{\bA}{F}^2\norm{\bB}{F}^2)$ where $\delta\ge 0.1$.
\end{theorem}

We reiterate that FFT was already used in the AMM line of work (TensorSketch \cite{pagh13}),  albeit in a different way. Our work is the first to surpass the linear $1/r$ speed-accuracy tradeoff \eqref{eq_AMM_sota} for AMM, using nothing but FFTs.

One suboptimal aspect of \autoref{thm_S&S_AMM} is that it is agnostic to the structure of the polynomials being convolved (the construction used in \autoref{thm_main_exactMM_informal}). To address this, we suggest designing low degree polynomials which satisfy an approximate version of the \textit{triple product property} \cite{CU03}, that captures when convolutions can embed matrix multiplication. In the abelian case, the approximate notion is equivalent to finding \textit{sumsets which minimize the (multi)-intersection with arithmetic progressions}. We prove a lower bound on this problem in the abelian case, and complement it with a simple tight construction in the group $\Z_{rn^2}$. In the non-abelian case, we extend a proof of Gowers to obtain a weaker lower bound.

We propose a natural family of randomized algorithms that achieve state-of-the-art\footnote{Up to a log factor in runtime.} approximation (linear speed-accuracy tradeoff), using a \textit{single} convolution in a small abelian group to compute an approximation of a \textit{single} matrix product (as opposed to \autoref{thm_S&S_AMM}, where each convolution encodes multiple matrix products, using the rich combinatorial structure of \textit{multi-variate} polynomials).

\begin{theorem}[AMM Convolution Lower Bound, Informal version of \autoref{thm:polyform-error-lower-bound}, \autoref{thm:upper-bound-polyform}]\label{thm_LB_abelian_AMM}
    Any (randomized) algorithm computing the product of two Rademacher matrices $\bA,\bB\in_R \set{\pm1}^{n\times n}$ using a single \textit{convolution} $\bC_G(\bA,\bB)$ in an abelian group $|G| \leq rn^2$, has expected error $\norm{\bC_G(\bA,\bB)-\bA\bB}{F}^2 \geq  \Omega(\frac{1}{r}\norm{\bA}{F}^2\norm{\bB}{F}^2)$. We provide a simple algorithm  (degree-$rn^2$ univariate polynomials) matching this bound.
\end{theorem}

\autoref{thm_LB_abelian_AMM} generalizes the $\Omega( n^3)$ lower bound of \cite{CU03} on the size of abelian groups realizing \emph{exact} matrix multiplication using convolutions, and provides an alternative, elementary algorithm achieving \eqref{eq_AMM_sota}, while possibly having better qualities such as high output rank. Note that this result does not rule out gaining improvements by computing multiple products in a single convolution (akin to \autoref{thm_S&S_AMM}). The upside of this approach is that it is non-asymptotic at all.

Finally, we suggest a generalization of this notion to the \textit{simultaneous} TPP setup, which captures a convolution's capacity to encode multiple independent matrix multiplications via multi-variate polynomials \cite{CKSU05}. We propose studying low-degree approximations of multi-variate polynomials, in hope of breaking the state-of-the-art while remaining non-asymptotic.

\subsection{Technical Overview}

Below we provide a high-level streamlined overview of the ideas behind the proofs of the main results.

\paragraph{Matrix Multiplication via Convolutions.}

The idea of using group convolutions for computing Matrix Multiplication was first presented in the seminal work of Cohn and Umans \cite{CU03}. Instead of directly multiplying $\bA,\bB\in\R^{n\times n}$, they proposed embedding the matrices as polynomials in $\R[x]_{/(x^{n^3}-1)}$ (identify $x^{n^3}$ with $1$), choosing a specific set of monomials. In particular, set $P_\bA(x)=\sum_{i,j\in [n]}\bA_{i,j}x^{-n^2i+nj}$ and $P_{\bB}(x)=\sum_{j,k\in [n]}\bB_{j,k}x^{k-nj}$ (all powers are taken modulo $n^3$). Let $P_{\bC}=P_\bA \cdot P_{\bB}$ denote the polynomial product. Suppose $P_{\bC}(x)=\sum_{s\in [n^3]}c_s x^s$, then it is easy to see that 
\begin{equation}\label{eq_circ_conv_MM}
c_{k-n^2i}=\sum_{j\in [n]}\bA_{i,j}\bB_{j,k}=(\bA\bB)_{i,k}.
\end{equation}
Hence, by computing a single polynomial multiplication via the \emph{Fast Fourier Transform} (FFT), which is equivalent to the cyclic convolution of the polynomials' coefficient vectors, one can recover the product matrix $\bA\bB$. This is true because the group-embedding used in \eqref{eq_circ_conv_MM} satisfies the property that, for any $(i,k)$, the equation $(k' -k) + n(j-j') + n^2 (i-i')=0$ has \emph{exactly} $n$ solutions, given by $i'=i,k'=k$ and $j'=j\in [n]$. In other words, only the desired \emph{``diagonal"} terms $\bA_{ij} \bB_{jk}$ in the convolution are aggregated to the monomial $x^{k-n^2i}$, whereas the undesired ``cross-terms" ($i'\not=i,k'\not=k, j\neq j'$) are aggregated into other ``garbage" monomials which we simply ignore when recovering $\bC$ from $c$. Cohn and Umans termed this property  the \textit{Triple Product Property} (TPP). 

Unfortunately, the convolution algorithm in \eqref{eq_circ_conv_MM} is quite useless, as the polynomial $P_\bA$ has degree $n^3$, hence the convolution (via FFT) requires $\Omega(n^3 \log n)$ time, slower than na\"ive matrix multiplication(!). 
As such, \cite{CU03} generalized the group-convolution approach \eqref{eq_circ_conv_MM} and the TPP property to \emph{general} groups, where instead of polynomials, the matrices $\bA$ and $\bB$ are embedded into the ``group algebra" $\C[G]$ of some general but fixed group $G$, and the two embeddings are then convolved \emph{inside $\C[G]$} to produce an element of $\C[G]$ containing the desired output $\bA\bB$.\footnote{More formally, embed the indices $\bA(i,j), \bB(j',k) \to G$ and compute 
$(\sum_{g \in G} \bA_g g) \cdot (\sum_{h \in G} \bB_h h)=\sum_{f \in G}(\sum_{g h=f} \bA_g \bB_h) f$.} The  TPP condition for general groups then becomes: 

\begin{definition}[TPP]\label{def:TPP}
We say that a finite group $G$ satisfies the \textit{TPP} for $\ip{n,m,p}$ if there exist sets $S,T,U\subseteq G$ such that $\abs{S}=n,\abs{T}=m,\abs{U}=p$ and for every $s_1,s_2\in S,t_1,t_2\in T,u_1,u_2\in U$ the equality $(s_1s_2^{-1})(t_{1}t_2^{-1})(u_1u_2^{-1})=1_G$ implies $s_1=s_2,t_1 =t_2,u_1=u_2 $.
\end{definition}

The sets $S,T,U$ play the role of indexing sets: $S$ is used to index the rows of $\bA$ and $\bA\bB$, $T$ the columns of $\bA$ and rows of $\bB$, and $U$ the columns of $\bB$ and $\bA\bB$. Another view of this definition, is that it captures when the \emph{tensor} of the group algebra of $G$ (which describes the convolution operator in $\C[G]$) 
can be \emph{restricted} to the matrix multiplication tensor \cite{Bla13}.

The aforementioned $\Omega(n^3)$ lower bound on the degree of TPP polynomials in \eqref{eq_circ_conv_MM} is not a coincidence --  A simple proof \cite{CU03} shows that  \textit{any abelian} group satisfying the TPP must have $\abs{G}\ge n^3$, which seems to make abelian groups useless.

The TPP lower bound for abelian groups shifted the focus of the group-theoretic-FMM line of work to finding \textit{non-abelian} TPP groups of sub-cubic size\footnote{In fact, an example of a group $G$ satisfying the TPP with $\abs{G}=n^{2+o(1)}$ was given in \cite{CU03}.} \cite{CU03,CKSU05}. Alas, non-abelian groups generally lack fast convolution algorithms, due to: \textbf{(i)} the absence of a fast Fourier transform algorithm; \textbf{(ii)} The fact that Fourier transform 
of non-abelian groups transforms convolutions into \emph{block-diagonal} matrix products, as opposed to element-wise product in the abelian case.
This is where \textit{divide-and-conquer} comes into play:  applying the algorithm recursively on \textit{tiles} (a-la Strassen) amortizes the cost of any linear transformation  done before or after the recursive step, so applying the Fourier transform of $G$ is asymptotically free of charge\footnote{This was explained in a previous footnote, discussing the computation of $\bU^{\otimes k}v$.}. This is also where our approach departs from the group-theoretic line of work -- \textbf{\emph{Since we aim for \textit{non-recursive} (``combinatorial'') algorithms, we must inevitably work with \textit{abelian} groups 
which have FFTs}}. While this seems like a dead-end in light of the $\Omega(n^3)$ degree lower bound for TPP polynomials, it turns out that \emph{multivariate} polynomials (e.g., convolutions in $\Z_m^k$) satisfy a more general TPP condition, proposed in \cite{CKSU05}, which allows to  
multiply independent \emph{block-matrices}, and thus circumvents the $\Omega(n^3)$ barrier. The property is defined by:

\begin{definition}[STPP]\label{def:STPP}
    We say that a finite group $G$ satisfies the \textit{Simultaneous TPP} for $\ip{n_1,m_1,p_1},\ldots \ip{n_\ell,m_\ell,p_\ell}$, if there exist triplets of subsets in $G$, $\set{(S_i,T_i,U_i)}_{i=1,\ldots,\ell}$, of sizes $\abs{S_i}=n_i,\abs{T_i}=m_i,\abs{U_i}=p_i$, such that $(S_i,T_i,U_i)$ satisfy the TPP, and they don't mix:
    $$\forall i,j,k\in\set{1,\ldots,\ell}:\quad s_i \cdot (s_j^\prime)^{-1}\cdot t_j \cdot (t_k^{\prime})^{-1}\cdot u_k\cdot (u_i^{\prime})^{-1}=1_G\implies i=j=k,$$assuming $x_a,x_a'\in X_a$ (with $X\in \set{S,T,U},a\in \set{i,j,k}$).
\end{definition}

Simply stated, the STPP is met when $G$ contains $\ell$ triplets $(S_i,T_i,U_i)$ which marginally satisfy the \nameref{def:TPP} and which don't mix together under the convolution of $G$. 

\subsubsection{Exact Matrix Multiplication (Proof of \autoref{thm_main_exactMM_informal})} 
A simple construction of \cite{CKSU05} gives a non-trivial abelian group that satisfies the \nameref{def:STPP}. In particular, they showed that the group $\Z_m^{3N}$ ($3N$-variate polynomials with degree at smaller than $m$ on each variable) contains $\ell=2^N$ triples of size $q=(m-1)^N$, allowing to compute $2^N$ independent matrix products of size $(m-1)^N\times (m-1)^N$ in a single convolution. In its most basic form, by setting $N=1$ (so $d=2,q=(m-1)$), it allows embedding two independent matrix products $\bA_1\times \bB_1$ and $\bA_2\times \bB_2$, of size $m-1\times m-1$ via the convolution of the following polynomials in $\Z_m^3$:
\begin{align}\label{eq:stpp-polynomials}
    \bA_1 = (\bA_{i,j}^{(1)})_{i,j=0}^{m-2},\ \bA_2 = (\bA_{i,j}^{(2)})_{i,j=0}^{m-2}\quad & \leadsto\quad  a(x,y,z)=\sum_{i,j=0}^{m-2}\para{\bA_{i,j}^{(1)}\cdot x^{i+1}y^{j+1}+\bA_{i,j}^{(2)}\cdot y^{i+1}z^{j+1}},
    \\ \bB_1 = (\bB_{i,j}^{(1)})_{i,j=0}^{m-2},\ \bB_2 = (\bB_{i,j}^{(2)})_{i,j=0}^{m-2}\quad & \leadsto\quad  b(x,y,z)=\sum_{i,j=0}^{m-2}\para{\bB_{i,j}^{(1)}\cdot y^{i+1}z^{j+1}+\bB_{i,j}^{(2)}\cdot x^{i+1}z^{j+1}}.\notag
\end{align}

Since \cite{CKSU05} aim for the best possible bound on $\omega$ (the matrix multiplication exponent), they invoke this algorithm recursively (done implicitly through the Asymptotic Sum Inequality \cite{Sch81}).
Our simple observation is that one can instead compute the polynomial multiplication \eqref{eq:stpp-polynomials} directly using the group's FFT (the group is $\Z_m^{3N}$) to compute $d$ products of $q\times q$ matrices, asymptotically faster than the na\"ive algorithm $O(d\cdot q^3)$ (albeit far from the optimal bound on $\omega$ achievable via recursion). This is illustrated in \autoref{fig_STPP}.

Given this black-box algorithm for computing $d$ products of $q\times q$ matrices, we can use it to multiply $dq\times dq$ square matrices $\bA,\bB$ by partitioning the matrices into $d\times d$ blocks of size $q\times q$, and invoke the algorithm $d^2$ times. This amplifies the runtime savings of the algorithm, and substitutes the role of recursion in the analysis, leading to the exponent in \autoref{thm_main_exactMM_informal}.

\begin{figure}[ht]
    \centering
    \includegraphics[width=0.75\linewidth]{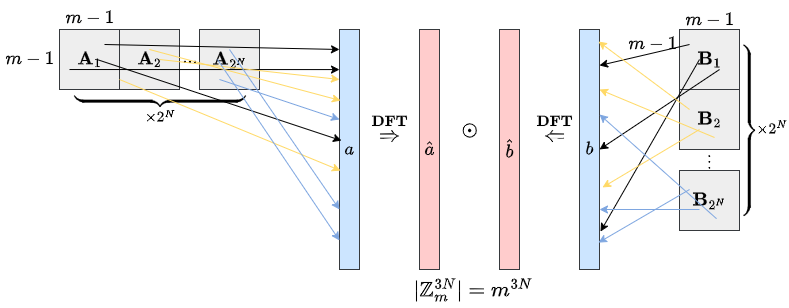}
    \caption{The convolution algorithm underlying \autoref{thm_main_exactMM_informal}. The group is $\Z_m^{3N}$. We are embedding $d=2^N$ square matrices of size $q = (m-1)^N$ in two polynomials $a,b$, computing their Fourier transforms, and computing the element-wise product. The decoding part is not presented in this figure, but is essentially the same. Invoking this algorithm $d^2$ times (possibly in parallel) leads to the fast combinatorial algorithm.}
    \label{fig_STPP}
\end{figure}

\subsubsection{Sketch and Solve Approximation (\autoref{thm_S&S_AMM})}

Using the algorithm obtained above, we observe that combining it with standard linear sketching techniques (\cite{sar06,drineas06,cw13}) of the input matrices in a black-box manner, already improves the best known asymptotic speed-accuracy tradeoff \eqref{eq_AMM_sota} achievable by combinatorial algorithms. 
The proof is straightforward: We first use a random Gaussian matrix $S\in_R \cN(0,1/\sqrt{n})^{n \times r}$ to ``compress" the number of columns (resp. rows) of the input matrices $\bA' := \bA S$, $\bB' := S^\top \bB$, yielding two $(n\times r)$ matrices. The JL moment property \cite{cw13} then guarantees that
\begin{equation}\label{eq_JL_AMM}
\E_{S}\| \bA\bB - \bA'\bB' \|_F^2 \leq \frac{1}{r}\norm{\bA}{F}^2\norm{\bB}{F}^2.
\end{equation}
It therefore remains to multiply $\bA'\bB'$, which using na\"ive \matmul\; requires $O(rn^2)$ time, recovering the standard linear speed-accuracy AMM error. Our obvious approach is to use \autoref{thm_main_exactMM_informal} (or more precisely, a rectangular version of it), to multiply $\bA'\bB'$ slightly faster, in $ \approx r^{0.9}n^2$ time. 
Note that the composed algorithm reduces the dimensions, and therefore the degree, of the STPP polynomials being convolved.
Since this yields a (small) polynomial saving in runtime, we can afford to increase the number of rows in the sketching matrix $S$ to $\approx r^{1.1}$, while still matching the original runtime budget $O(rn^2)$. 
By \eqref{eq_JL_AMM} this yields an error of $\approx \frac{1}{r^{1.1}}\norm{\bA}{F}^2\norm{\bB}{F}^2$. \\ 

We remark that Pagh's TensorSketch algorithm \cite{pagh13} also uses FFT (albeit in a different way, to efficiently recover the ``heavy-hitters" of an \emph{outer product} $uv^\top \in \R^{n\times n}$), but nevertheless his algorithm still yields only a linear dependence on $r$, as in \eqref{eq_AMM_sota}. 

\subsubsection{Approximate Triple Product Polynomials}

One natural approach, is to relax the definition of the \nameref{def:TPP}, by allowing non-trivial solutions. We define $\rho_n(G)$ to be the minimal number of solutions to the equation $$s^{-1}s't^{-1}t'u^{-1}u'=1_G$$over all indexing sets $S,T,U\subset G$ (which have enough structure to allow embedding $n\times n$ matrices). In other words, $\rho_n(G)=\min_{(S,T,U)}1_{S^{-1}}*1_{S}*1_{T^{-1}}*1_T*1_{U^{-1}}*1_U(e_G)$. We also extend this definition to the \nameref{def:STPP} case. We note that the case of \textit{abelian} groups, this quantity is minimized by highly structured additive sets, namely arithmetic progressions. In particular, it can be re-written as $$\min_{(S,T,U)}\norm{[(S-U)-(S-U)]\cap (T-T)}{},$$where $\norm{\cdot }{}$ denotes multi-set cardinality. It is easy to show that in the group $\Z_{rn^2}$, this quantity is minimized when $S-U$ is an arithmetic progression, and $T$ is chosen so that the difference set $T-T$ avoids the arithmetic progression of $S-U$.

We show that for a family of randomized algorithms, based on convolving the embedded signals in $G$ after adding some random signs and permuting the indexing sets, the error is related to the quantity $\rho_n(G)$ in the following way. We show that the \textit{normalized} error is lower bounded by $\Omega\para{\frac{\rho_n(G)-n^3}{n^4} }$, which we also show is tight. The main takeaway of \autoref{thm_LB_abelian_AMM} is that $\rho_n(G)$ precisely determines the speed-accuracy tradeoff of the randomized algorithm.

Clearly, the size of $\rho_n(G)$ is inherently related to the group structure, and one could hope to find some special sets which successfully achieve a very low number of non-trivial solutions, compared to the group size, thus leading to a super-linear tradeoff. In abelian groups, one can transform this question into a finding sumsets which minimize the multi-intersection with an arithmetic progression. However, a nifty Cauchy-Schwarz argument shows that $\rho_n(G)\ge \frac{n^6}{\abs{G}}$, implying that abelian groups cannot break the linear tradeoff. We provide a very simple construction, based on arithmetic progressions in $\Z_{rn^2}$, (almost) attaining the lower bound.

For non-abelian groups, the lower bound is a bit more subtle, and requires Fourier analytic machinery. We provide a proof -- which is directly based on a proof of Gowers for mixing in non-abelian groups (see \cite{Bre14}) -- showing that $\rho_n(G)\ge \frac{n^6}{\abs{G}}-n^3\cdot \sqrt{\frac{\abs{G}}{d(G)}}$, where $d(G)$ is the smallest degree of a non-linear complex representation of $G$. This is a much weaker lower bound, as it requires $d(G)\approx r^3$, assuming $\abs{G}\le rn^2$, for it to have any meaningful interpretation. The non-abelian case is less interesting for practical algorithms (due to the lack of FFT and the block-diagonal form), but may be of independent interest. This result leaves the case of small $d(G)$ open.

\section{Preliminaries}
\subsection{Constructions}\label{sec:constructions}
\paragraph{Vanilla TPP.}\label{subsec:vanilla-tpp}
Let $n,m,p$ be three given integers. Define $G=\Z_{nmp}$, and define $$S_{n,m,p}=\set{x\cdot mp\mid x\in [n]}\quad,\quad T_{n,m,p}= \set{x\cdot p\mid x\in [m]}\quad,\quad U_{n,m,p}=\set{x\mid x\in [p]}.$$Then $(S_{n,m,p},T_{n,m,p},U_{n,m,p})$ realize the \nameref{def:TPP} for $\ip{n,m,p}$ in $G$. Let $a,b$ denote the embedded signals of $\bA,\bB$ respectively, then $$\supp(a)=T-S=\set{xmp+yp\mid x\in [n],y\in [m]}\quad,\quad \supp(b)=U-T=\set{y-xp\mid x\in [m],y\in [p]},$$showing that $\supp(a)$ is an arithmetic progression (AP) with difference $p$ starting at $0$, while $\supp(b)$ is just the interval $([nmp-(m-1)p,nmp-1]\cup[0,p-1])\cap \Z$ (this is a continuous interval because we work on the cyclic group). When $n=m=p$, we denote $S_{n,m,p}=S_n$ and similarly for $T,U$.

\paragraph{CKSU STPP.}\label{subsec:vanilla-stpp} This construction was first given in \cite{CKSU05}. Set $H=\Z_{m}^3$ and denote each copy of $\Z_m$ in $H$ by $H_1,H_2,H_3$ (meaning $H_1=\mathrm{Im}(\phi_1)$ where $\phi_1:\Z_m\hookrightarrow H$, $x\mapsto (x,0,0)$, and similarly $H_2,H_3$ are defined). Define 
\begin{align*}
S_{1}^{0} & =H_{1}\setminus\left\{ 0_{H}\right\}  & S_{2}^{0} & =H_{2}\setminus\left\{ 0_{H}\right\}  & S_{3}^{0} & =H_{3}\setminus\left\{ 0_{H}\right\} \\
S_{1}^{1} & =H_{2}\setminus\left\{ 0_{H}\right\}  & S_{2}^{1} & =H_{3}\setminus\left\{ 0_{H}\right\}  & S_{3}^{1} & =H_{1}\setminus\left\{ 0_{H}\right\} 
\end{align*}

The authors in \cite{CKSU05} have shown that this sets satisfy the \nameref{def:STPP}. It allows to compute $2$ independent products of $(m-1)\times (m-1)$ matrices. Using \autoref{lem:tpp-product-property}, we set $G=H^N$ for some $N\ge 1$, and by a single convolution in $G$ we compute $2^N$ independent products of $(m-1)^N\times (m-1)^N$ matrices. The STPP sets can be indexed by a binary string $\ba\in \set{0,1}^N$: $$(i\in \set{1,2,3})\qquad S_i^{\ba}=S_i^{\ba_1}\times \ldots \times S_i^{\ba_N}.$$

\section{Sub-Cubic Matrix Multiplication via Convolutions}\label{sec:fast-mm}
Given an algorithm that computes $k$ products of $n\times n$ matrices, we can use it to multiply $nk\times nk$ matrices $\bA,\bB$, by viewing them as $k\times k$ matrices where each entry is a block of size $n\times n$. Then we need $k^3$ block products, which we sum up in a specific order. Thus, it can be computed through $k^2$ invocations of the given algorithm. More precisely, let $\bA_{i,j},\bB_{i,j}$ denote $(i,j)$-th block (so $i,j\in [k]$), then to compute $(\bA\bB)_{i,j}$ we compute the products $\set{\bA_{i,t}\bB_{t,j}\mid t\in [k]}$ and sum up the results.

The STPP construction \nameref{subsec:vanilla-stpp} gives rise to a fast algorithm that computes $k=2^N$ products of $(m-1)^N\times (m-1)^N$ matrices, while costing a \textit{single} convolution in the group $\Z_{m}^{3N}=H^N=G$. Since this group is abelian, it has an FFT, which costs $O(m^{3N}\log m^{3N})$. For completeness, we give a full description of the algorithm in \autoref{app:fast-cmm-algorithm}.

\begin{remark}
    Using other STPP constructions, like the CW tensor \cite{cw87}, may lead to faster asymptotic algorithms using the same scheme described above. However, the basic construction usually shows up in a very high power of the tensor (as in the laser method \cite{AW21}).
\end{remark}
\subsubsection*{Runtime Analysis}

For a square matrix of size $n_{0}\times n_{0}$ where $n_{0}=\left(m-1\right)^{N}\cdot2^{N}$, we can use the notation of the blocking scheme as follows: define $n=\left(m-1\right)^{N}$ and $k=2^{N}$. Write $$\tau=\frac{3\log(m)}{\log (m-1)}\quad,\quad \eta = \frac{\log(2)}{\log(m-1)}.$$Then $n^{\tau}=m^{3N}$, $n^{\eta}=k$, and so $n_0 =nk=n^{1+\eta}$. Therefore the runtime of the algorithm is $$O(k^2 m^{3N}\log m^{3N})=O(n^{\tau+2\eta}\log n)=O\para{n_0 ^{\frac{\tau+2\eta}{1+\eta}}\log n_0}.$$

For example, setting $m=8$ we obtain $\frac{\tau+2\eta}{1+\eta}\approx 2.89$, and optimizing the expression over $m$, i.e., optimizing $f(x)=\frac{3\log(x)+2\log 2}{\log(x-1)+\log 2}$, we see that it has a local minima at $x\approx 20.23$ for which $f(x)\approx 2.85$, and $\lim_{x\to \infty} f(x)= 3$, so the best possible algorithm from this approach has exponent $\approx 2.85$.

\subsubsection*{Runtime Asymptotics}

Suppose the algorithm's runtime is $C\cdot m^{3N}\log(m^{3N})$ while the naive implementation runtime is $2\cdot 2^N\cdot (m-1)^{3N}-2^N\cdot (m-1)^{2N}$ (the number of FLOPs). In the table below we give some estimates for the threshold where our algorithm becomes faster, assuming different values of $C$ and $m$. We note that a meticulous calculation of the exact FLOP count for the algorithm would possibly result in bigger constants, which has great effect on the threshold point. In other words, this is merely a sketch of the runtime asymptotics, which shows how sensitive the threshold point is to the underlying constants. Moreover, at these sizes the impracticality of the algorithm stems from the enormous memory usage ($m^{3N}$ for each element in the convolution).

\begin{table}[ht]
    \centering
    \begin{tabular}{CCCG}
        \hline
        m & C & \text{minimal value of $N$} & \text{total matrix size }(n) \\\hline
        6 & 2 & 30 & 10^{30} \\\hline
        10 & 2 & 10 & 18^{10} \\\hline
        10 & 5 & 12 & 18^{12} \\\hline
        8 & 5 & 16 & 14^{16} \\\hline
    \end{tabular}
\end{table}

\section{Sketch \& Solve Approximation}\label{sec:sns}

The first approach we will present for solving the approximation problem is to reduce the runtime complexity by sketching the input and then applying the algorithm presented in \autoref{sec:fast-mm}. This is a ``black-box'' approach, combining known techniques to obtain a slight improvement. Instead of applying the algorithm on the input matrices, we first sketch the input matrices to lower rank matrices and then compute their product using the algorithm. Note that while the algorithm was presented for square matrices, it can be easily adapted to rectangular matrices, by changing the partitioning scheme. Although this presents no major or practical advancements, it is worth laying out explicitly.

A simple formulation of sketching results (see \cite{woodruff2014sketching} and references within) for our needs is the following:
\begin{lemma}[JL-Sketch]\label{lem:jl-sketch}
    Set $m=\frac{1}{\eps^2}$, draw $S\in \R^{n\times m}$ where each coordinate is a random normal Gaussian scaled by $\frac{1}{\sqrt{m}}$. Then $$\E[\norm{\bA \bB-\bA SS^{\top}\bB}{F}^2]=O(\eps^2 \cdot \norm{\bA}{F}^2\cdot \norm{\bB}{F}^2).$$In fact, this inequality holds with high probability.
\end{lemma}

In particular, in order to reach an error of $O(\frac{1}{r}\norm{\bA}{F}^2\norm{\bB}{F}^2)$ we need $m= r$. Computing $\bA S$ and $S^{\top}B$ takes $rn^2$ using the naive implementation, and so does computing their product. Note that this process inherently reduces the rank of the output matrix in order to achieve faster computation.

\begin{theorem}\label{thm:sketch-n-solve}
    Suppose we are given two square matrices $\bA,\bB$ of size $n$ and a parameter $r\le n$ which is not insignificant (compared to $n$). If $n$ is large enough, there exists a sketching algorithm that approximates the multiplication of $\bA,\bB$, by computing a matrix $\bC$ in time $O\left(n^2\cdot r^{\delta + o(1)}\right)$ where $\delta\le0.9$ and $r$ bounds the rank of the matrix $\bC$ while also satisfying $$\ex{}{\norm{\bA\bB-\bC}{F}^2}\le O\left(\frac{1}{r} \norm{\bA}{F}^2\norm{\bB}{F}^2\right).$$In fact, this inequality holds with high probability.
\end{theorem}
By saying that $n$ and $r$ are large enough, we mean that the combinatorial STPP algorithm achieves savings for $r\times r$ matrices.
\begin{proof}
    Fix $m$ and $N$ such that $r=2^N(m-1)^N$. Write $p=\lceil n/(m-1)^N\rceil$ (so $p\ge 2^N$ because $n\ge r$). Viewing $S$ as a $p\times 2^N$ block matrix, and $\bA,\bB$ as $p\times p$ block matrices (padding them if $p>n/(m-1)^N$), the products $\bA S,S^{\top}\bB$ and $(\bA S)(S^{\top}\bB)$ can all be computed by $p^2\cdot 2^N$ block products. In each invocation of the STPP algorithm we compute $2^N$ block products, so the overall number of algorithm invocations is $p^2$. The runtime of each invocation is $O(3N m^{3N}\log m)$. Thus our overall runtime is $$O(p^2\cdot 3N\cdot m^{3N}\log m)=O\left(\frac{n^2\cdot 3N\cdot m^{3N}\log m}{(m-1)^{2N}} \right).$$Denoting $s=\frac{3N\cdot m^{3N}\log m}{(m-1)^{2N}}$ we observe that ($\log$ is base $2$) $$\log_{r}(s)= \frac{\log(3N)}{N(1+\log(m-1))}+\frac{3\log(m)-2\log(m-1)}{1+\log(m-1)}.$$The dominating part is $f(m)=\frac{3\log m-2\log(m-1)}{1+\log (m-1)}$ that approaches $1$ from below as $m\to \infty$ and has a minimum at $m\approx 20$ (for which $f(m)\approx 0.85$). For $m\in [8,480]$, we have $f(m)\le 0.9$. The second part $g(N)=\frac{\log(3N)}{N(1+\log(m-1))}$ vanishes with $N$. Therefore, we can compute the sketch in time $O(n^2r^{\delta+o(1)})$, where $\delta\le 0.9$ (we can always use small $m$ and take $N$ to be large). The claim then follows from the JL sketch theorem.
\end{proof}

This theorem shows that for large enough matrices, we can achieve a better bound on the error, in the same runtime, thus beating the state of the art. 

\section{Approximate Triple Product Polynomials}\label{sec:polyform}
The second approach we will present is a non-black box approach. The idea is to embed the matrices into sets which are approximately TPP/STPP sets. For the sake of simplicity of presentation we will focus on TPP. Relaxing the TPP condition allows for a smaller group size, and assuming the group has an FFT algorithm, we obtain a fast algorithm to compute an approximate matrix multiplication.

\subsection{Approximate Triple Product Quantity}
We work with square matrix multiplication, although the case of rectangular matrices is very similar, up to some minor adjustments.

\begin{definition}[Indexing Triplet]\label{def:indexing-triplet}
    Let $G$ denote a group with size at least $n^2$.  We say that three subsets $S,T,U\subseteq G$ are \textit{an indexing triplet} for $G$ and $n$, if $\abs{S}=\abs{T}=\abs{U}=n$ and $\abs{S^{-1}T}=\abs{T^{-1}U}=n^{2}$. We denote the collection of indexing triplets by $\cI(G,n)$.
\end{definition}

\begin{definition}[Approximate Triple Product Quantity]\label{def:atpq} Let $G$ be a group of size at least $n^2$, then for any indexing triplet $(S,T,U)\in \cI(G,n)$ we define the \textit{approximate triple product quantity} for $(S,T,U)$, denoted by $\rho_{n,G}(S,T,U)$, as the number of solutions to the equation $$s_1 s_2^{-1} t_1 t_2^{-1}= u_1 u_2^{-1}\qquad,s_1,s_2\in S,t_1,t_2\in T, u_1,u_2\in U.$$Equivalently, $\rho_{n,G}(S,T,U)=1_S*1_{S^{-1}}*1_T*1_{T^{-1}}*1_U*1_{U^{-1}}(1_G)$, where $1_X$ is the indicator function of a set $X\subset G$. We also define $$\rho_n(G)=\min_{(S,T,U)\in \cI(G,n)} \rho_{n,G}(S,T,U),$$i.e. the minimal number of solutions to the equation, and call this quantity the \emph{Approximate Triple Product Quantity} of $G$ with respect to $n$.
\end{definition}

Note that $\rho_n(G)\ge n^3$, and equality holds if and only if $G$ satisfies the TPP. This follows from the fact that there are always $n^3$ trivial solutions to the equation, obtained by setting $s_1 =s_2,t_1=t_2$ and $u_1=u_2$. Furthermore, a simple argument shows that for an abelian group $G$ to satisfy the TPP for size $n$, we must have $\abs{G}\ge n^3$ (to see this, note that $\phi:S\times T\times U\to G$ given by $(s,t,u)\mapsto s+t+u$ is an injection, using the TPP and the commutativity of $G$).

\begin{lemma}\label{lem:atpq-lower-bound}
    Let $G$ be a finite group, $A,B\subset G$, then $$\left|\left\{(a_1,a_2,b_1,b_2)\in A^2\times B^2 \mid a_1 a_2^{-1} =b_1 b_2^{-1} \right\}\right| \ge \frac{\left|A\right|^2 \cdot \left|B\right|^2}{\left|G\right|}.$$
\end{lemma}
\begin{proof}
    Note that $a_1 a_2^{-1}= b_1 b_2^{-1}$ if and only if $b^{-1}_1 a_1=b_2^{-1}a_2$, so we can rewrite the LHS by -
    \begin{align}\label{atpq-LHS}
        \sum_{g\in G}& \left| \left\{ (a_1,a_2,b_1,b_2)\in A^2 \times B^2 \mid b_1^{-1} a_1=g = b_2^{-1}a_2\right\}\right| = \sum_{g\in G} \left|\left\{(a,b)\in A\times B\mid b^{-1}a=g\right\}\right|^2.
    \end{align}
By Cauchy-Schwarz we have $$\eqref{atpq-LHS}\ge \frac{1}{\left|G\right|} \left(\sum_{g\in G} \left|\left\{ (a,b)\in A\times B\mid b^{-1}a=g\right\}\right|\right)^{2}.$$
The proof is complete by noting $\sum_{g\in G} \left|\left\{ (a,b)\in A\times B\mid b^{-1}a=g\right\}\right|= \left|A\right|\cdot \left|S\right|$.
\end{proof}

\begin{corollary}\label{cor:diffset_LB}
    Let $G$ be a finite \textbf{abelian} group, then the \nameref{def:atpq} satisfies $$\rho_{n}(G)\ge \frac{n^6}{\left|G\right|}.$$
\end{corollary}
\begin{proof}
Let $S,T,U\subset G$ be subsets of size $n$. Recall that $\rho_{n,G}(S,T,U)$ is the number of solutions to the following equation, using additive notation, $$s_{1}-s_{2}+ t_{1}-t_{2}+u_{1}-u_{2}=0,$$which is equivalent to (by commutativity), $$s_{1}-s_{2} = (t_{2}-u_{1})- (t_{1}-u_{2}).$$Denote $B=TU^{-1}=T-U=\{t-u\mid t\in T,u\in U\}$, and apply \autoref{lem:atpq-lower-bound} on $A=S$ and $B$ to obtain that the number of solutions to the last equation (equal to $\rho_{n,G}(S,T,U)$) is at least $\abs{S}^{2}\cdot \abs{T^{-1}U}^{2} / \abs{G}$. We finish by noting that $n^{2}= \abs{T^{-1}U}=\abs{ TU^{-1}}$ by definition of an indexing set (here we use the fact $G$ is abelian).
\end{proof}

For an abelian group with $\abs{G} = rn^{2}$, we conclude that $\rho_n(G)\ge \frac{n^4}{r}$. Note that for the inequality to be an equality, we need $\abs{\set{(s,t,u)\in S\times T\times U:s+t-u=g}}$ to be a constant for every $g\in G$. For $r<n$ the constant is larger than $1$, and so $(S,T,U)$ cannot be an indexing triplet, at least when $G=\Z_{rn^2}$ (for $1_S*1_T*1_{U^{-1}}$ to be constant, $S,T,U$ must all reside in subgroups of $\Z_{rn^2}$, hence they must be arithmetic progressions that can't avoid each other, meaning that $\abs{S^{-1}T}\not=n^2$, in contradiction with the definition of indexing sets). We conclude that the lower bound is sharp for $r>1$, when $G=\Z_{rn^2}$.

Note that the corollary doesn't hold for non-abelian groups, because the proof uses commutativity to obtain the expression used in \autoref{lem:atpq-lower-bound}. For non-abelian groups we are able to prove a weaker statement, which depends on the smallest dimension of a non-linear representation of $G$. The proof follows Gowers' lemma (see \cite{Bre14}, Lemma 2.2) and is given in the appendix for completeness.

\begin{lemma}
    Let $d_{\min} = \min\set{d_{\pi}\mid \pi\in \mathrm{Irr}(G), d_{\pi}>1}$. Then for an indexing triplet $S,T,U\subset G$ we have $$\rho_{n,G}(S,T,U)\ge \frac{(\left|S\right|\left|T\right|\left|U\right|)^{2}}{\left|G\right|}-\left|S\right|\left|T\right|\left|U\right|\cdot \sqrt{\frac{\left|G\right|}{d_{\min}}}.$$
\end{lemma}

Assuming $\abs{G}\le rn^2$, in order for the lower bound to be non-trivial, we need $d_{\min}>\frac{n^{2}r^{3}}{(n-r)^{2}}\approx r^3$. In particular, this lower bound suggests that groups for which $d_{\min}$ is very large, the quantity becomes comparable with that of abelian groups (there is strong mixing).
\subsection{PolyForm Algorithm}

Given an indexing triplet $S,T,U\subset G$, we index $\bA,\bB$ using these sets (while implicitly defining an order on $S,T,U$) and embed the matrices into $\C[G]$, by setting $a_{s^{-1}t}=\bA_{s,t}$ and $b_{t^{-1}u}=\bB_{t,u}$. We then compute their convolution $c=a*b$ and decode $\bC_{s,u}=c_{s^{-1}u}$, which takes $O(\abs{G}\log \abs{G})$ assuming $G$ has an FFT algorithm (for example, if $G$ is abelian). The \nameref{def:atpq} determines the number of noise terms, which we call ``\textit{collisions}'', since $$c_{s_1^{-1}u_1}=\sum_{s_2,t_1,t_2,u_2:\ s_{2}^{-1}t_1t_2^{-1}u_2=s_1^{-1}u_1}a_{s^{-1}_2t_1}b_{t_2^{-1}u_2},$$and a term is summed into $\bC$ if and only if it is a solution to the equation $s_1s_2^{-1}t_1t_2^{-1}=u_1u_2^{-1}$. There are $n^3$ \textbf{valid} solutions, which are $s_2 =s_1,t_2=t_1,u_2 =u_1$, which encode the matrix multiplication, while all other terms are \textbf{noise terms}.

The actual Frobenius error of this algorithm then depends on the squared sum of the \textit{coefficients} corresponding to these error terms. Since the collisions are pre-determined by the indexing triplet, they might be skewed, meaning that certain positions ($s^{-1}u$) might suffer from relatively large numbers of collisions. Moreover, big error coefficients might end up together under certain orderings of $S,T,U$, but not under different orderings. Thus we should be mindful of how skewed a certain indexing triplet is. Moreover, to mitigate the latter effect, we can randomly choose the ordering of the indexing sets.

Another randomness we add to the algorithm is structured random signs to $\bA,\bB$. The goal here is to promote error terms cancellations, as we'll show next. Take $\ba,\bb,\bg$ to be independent random vectors $\set{\pm1}^n$ which are drawn uniformly at random. Define $\tilde{\bA}=(\ba\bg^{\top})\odot \bA$ and $\tilde{\bB}=(\bg \bb^{\top})\odot \bB$, where $\odot$ denotes the Hadamard product of matrices. Then $\tilde{\bA}_{s,t}=\ba_{s}\bg_{t}\cdot \bA_{s,t}$ and $\tilde{\bB}_{t,u}=\bg_{t}\bb_{u}\bB_{t,u}$. Since $\ba,\bb,\bg$ are independent, we have that $$\E[(\tilde{\bA}_{s,t}\cdot \tilde{\bB}_{t',u})^2]=\begin{cases}(\bA_{s,t}\cdot \bB_{t,u})^2 & t=t',\\ 0 & \text{else.}\end{cases}$$This is exactly what we are interested in -- canceling collisions which can be characterized by the condition $t\not=t'$. The convolution yields a matrix $\ti{\bC}$ which we use to obtain the final approximation of $\bA\bB$, simply by multiplying it with $\ba^{\top}\bb$ from the left.

This defines a randomized oblivious algorithm, which we denote by $\mathrm{RPF}$ (stands for Randomized PolyForm). Our goal is to understand the worst case error of this algorithm, meaning $$\cE=\max_{\bA,\bB\not=0}\set{\frac{\E\left[\norm{\mathrm{RPF}(\bA,\bB)-\bA\bB}{F}^2\right]}{\norm{\bA}{F}^2\cdot \norm{\bB}{F}^2}}.$$Note that the algorithm implicitly uses some indexing triplet, thus we define $\cE_G=\min_{(S,T,U)\in \cI(G,n)}\cE_{S,T,U}$ as the best algorithm error achieved by the group $G$. 

\subsection{Matching lower bound for convolutions in small abelian groups}

In this section we prove the lower bound in \autoref{thm_LB_abelian_AMM}, which generalizes the Cohn and Umans' \cite{CU03} lower bound of $\abs{G}\ge n^3$ by directly embedding matrix multiplication directly into a convolution.

\begin{theorem}\label{thm:polyform-error-lower-bound}
    If $G$ is a finite group with $\abs{G}\ge n^2$, then $\cE_{G}=\Omega\para{\frac{\rho_n(G)-n^3}{n^4}}$. Consequently, in abelian groups of size $\abs{G}\le rn^2$ the normalized error is $\Omega(1/r)$.
\end{theorem}

For the proof we'll introduce the following useful notations and definitions.
\begin{definition}[Collisions]\label{def:collisions}
    Let $G$ denote a finite group, $\abs{G}\ge n^{2}$ and $(S,T,U)\in \cI(G,n)$ and indexing triplet. We define the set of \emph{Collisions} with respect to $G,n$ and \nameref{def:indexing-triplet} $(S,T,U)$ by: $$\error{S,T,U}{}:=\set{(s,t,t^{\prime},u)\mid t\not=t^{\prime},\  \exists x\in S,y\in U:\ s^{-1}t{t^{\prime}}^{-1}u=x^{-1}y}.$$In words, elements of $\error{S,T,U}{}$ correspond with coefficients that would be wrongfully mapped into an index of interest (i.e. used to index an element of $\bA\bB$).
\end{definition}

\begin{definition}[Collisions \& Buckets]\label{def:buckets}
In the same conditions as above, we define the set of \textit{Collisions} for an index $x,y$, where $x\in S,y\in U$ by:$$\error{S,T,U}{x,y}= \set{(s,t,t^{\prime},u)\mid t\not=t^{\prime},s^{-1}t{t^{\prime}}^{-1}u=x^{-1}y}.$$We also define the \textit{Bucket} of $x,y$ by: $$\buck{S,T,U}{x,y}= \set{(s,t,t^{\prime},u)\mid s^{-1}t{t^{\prime}}^{-1}u=x^{-1}y}.$$
\end{definition}

We note the following useful characterization of the error.

\begin{lemma}\label{lem:error-characterization}
    Let $G$ denote a finite group, $\abs{G}\ge n^{2}$, $\bA,\bB\in \R^{n}$ be two given matrices and $(S,T,U)\in \cI(G,n)$ some \nameref{def:indexing-triplet}. Let $\bC$ denote the output of the deterministic PolyForm algorithm. Then for every $x\in S,y\in U$ we have $$\bC_{x,y}=\sum_{(s,t,t^{\prime},u)\in \buck{S,T,U}{x,y}} \bA_{s,t}\bB_{t^{\prime}u},$$and the error is $$\norm{\bC-\bA\bB}{F}^2 = \sum_{x\in S,y\in U}\left(\sum_{(s,t,t^{\prime},u)\in \error{S,T,U}{x,y}} \bA_{s,t} \bB_{t^{\prime},u}\right)^{2}.$$
\end{lemma}

\begin{proof}
    Follows from opening up the definition of a group algebra product, noting that $$(\bA \bB)_{x,y}= \sum_{\substack{s\in S, t\in T,u\in U \\ s^{-1}(tt^{-1})u=x^{-1}y}} \bA_{s,t} \bB_{t,u},$$leaving us with a sum over $(s,t,t^{\prime},u)\in \error{S,T,U}{x,y}$ in the coordinate $x,y$ of $\bC -\bA\bB$.
\end{proof}
\begin{proof}[Proof of \autoref{thm:polyform-error-lower-bound}]
    For brevity, we let $\cR$ denote a randomized PolyForm setup, composed of an indexing triplet $(S,T,U)\in \cI(G,n)$, permutations of the indexing sets and sign vectors $\ba,\bb,\bg$. We denote the deterministic error from this setup by $\cE_{\cR}(\bA,\bB)=\frac{\norm{\bC-\bA\bB}{F}^2}{\norm{\bA}{F}^2\cdot \norm{\bB}{F}^2}$, where $\bC$ is the output matrix of the algorithm with setup $\cR$.
    By Yao's principle, for any distribution $\cD$ over inputs of the algorithm (i.e. couples of matrices), it holds $$\cE_G\ge \min_{\cR}\ex{(\bA,\bB)\sim \cD}{\cE_\cR(\bA,\bB)}.$$
    Set $\mathcal{D}$ to be the uniform distribution on Rademacher matrices (all entries are independently drawn uniformly from $\{\pm 1\}$). We can assume without loss of generality that the permutations are all the identity permutations, because $\cD$ is invariant under permutation. Moreover, we know $\norm{\bA}{F}^2=\norm{\bB}{F}^2=n^2$ are constant for $(A,B)\sim \cD$. Thus, by \autoref{lem:error-characterization} and the fact that $(\bA\odot \ba\bg^{\top})_{s,t}=\ba_{s}\bg_{t}\bA_{s,t}$ (and similarly for $\bB\odot \bg\bb^{\top}$), we have 
    \begin{align*}
        n^4\cdot \ex{(\bA,\bB)\sim \cD}{\cE_{\cR}(\bA,\bB)} &= \ex{}{\sum_{x\in S,y\in U}\left(\sum_{(s,t,t^{\prime},u)\in \error{\cR}{x,y}}\ba_{s}\bg_{t}\bg_{t^{\prime}} \bb_{u}\cdot \bA_{s,t}\bB_{t^{\prime},u}\right)^{2}} \\
        &= \sum_{x\in S,y\in U}\ex{}{\left(\sum_{(s,t,t^{\prime},u)\in \error{\cR}{x,y}}\ba_{s}\bg_{t}\bg_{t^{\prime}} \bb_{u}\cdot \bA_{s,t}\bB_{t^{\prime},u}\right)^{2}}.\tag{1}
    \end{align*}
    
    Fixing $x\in S,y\in U$, we note that \begin{align*}
        & \ex{}{\left(\sum_{(s,t,t^{\prime},u)\in \error{\cR}{x,y}}\ba_{s}\bg_{t}\bg_{t^{\prime}} \bb_{u}\cdot \bA_{s,t}\bB_{t^{\prime},u}\right)^{2}} \\
        &=  \sum_{(s,t,t',u)\in \error{\cR}{x,y}}\mathbb{E}\left[ \left(\ba_{s}\bg_{t}\bg_{t^{\prime}} \bb_{u}\cdot \bA_{s,t}\bB_{t^{\prime},u}\right)^{2}\right] \tag{2}\\ 
        & + \sum_{\substack{(s,t,t',u) \\ \not= \\ (p,r,r',q)}\in \error{\cR}{x,y}} \mathbb{E}\left[\left(\ba_{s}\bg_{t}\bg_{t'}\bb_{u}\cdot \bA_{s,t} \bB_{t^{\prime},u}\right)\cdot \left(\ba_{p}\bg_{r}\bg_{r'}\bb_{q}\cdot \bA_{p,r} \bB_{r^{\prime}, q}\right)\right].\tag{3}
    \end{align*}
    Since $\ba,\bb,\bg$ and $\bA,\bB$ all take values in $\pm 1$, we obtain $$(2)=\sum_{(s,t,t',u)\in \error{\cR}{x,y}}\mathbb{E}\left[1\right]= \abs{\error{\cR}{x,y}}.$$
    Regarding $(3)$, the assumption $(s,t,t',u)\not=(p,r,r',q)$ implies that $(s,t)\not= (p,r)$ \textbf{or} $(t',u)\not=(r',q)$, hence at most one couple of random variables in the set $\{\bA_{s,t}, \bA_{p,r},\bB_{t',u}, \bB_{r',q}\}$ are not independent. Using the multiplicity of the expectation, we obtain $$(3)=\sum_{\substack{(s,t,t',u) \\ \not= \\ (p,r,r',q)}\in \error{\cR}{x,y}} \mathbb{E}\left[\ba_{s}\bg_{t}\bg_{t'}\bb_{u}\cdot \ba_{p}\bg_{r}\bg_{r'}\bb_{q}\right] \cdot \mathbb{E}\left[\bA_{s,t} \bB_{t^{\prime},u}\cdot \bA_{p,r} \bB_{r^{\prime}, q}\right]=0,$$where we used the fact that a multiple of independent Rademacher RVs ($\pm 1$ with equal probability) is $0$ in expectation.

    To sum up, we have shown $$(1)= \sum_{x\in S,y\in U} \left|\error{\cR}{x,y}\right|= \left|\error{\cR}{}\right|\ge \rho_n(G)-n^3.$$Since $\cR$ was arbitrary, we have thus shown that $$\cE_G\ge \min_{\cR}\ex{(\bA,\bB)\sim \cD}{\cE_{\cR(\bA,\bB)}}\ge \frac{\rho_{n}(G)-n^{3}}{n^{4}}.$$
\end{proof}

\subsection{A Matching Upper Bound}

In this section we prove the upper bound part in \autoref{thm_LB_abelian_AMM}, which states that $\rho_n(G)$ gives an upper bound on the normalized error, assuming a very mild condition on the indexing sets.

\begin{theorem}\label{thm:upper-bound-polyform}
    Let $G$ denote a finite group, with $(S,T,U)$ an indexing triplet achieving minimal collisions, with the condition that $TT^{-1}$ doesn't contain elements of order two aside from $0$. The randomized algorithm achieves expected normalized error of $\cE_G\le \frac{\rho_n(G)-n^3}{n^4-n^3}$.
\end{theorem}

\begin{proof}
    By randomly choosing an order for the sets $S,T,U$, we essentially choose bijective functions $\varphi_{S}:[n]\to S$ and similarly $\varphi_{T},\varphi_{U}$. Denote $\varphi=\varphi_{S}\otimes \varphi_{T}\otimes \varphi_{T} \otimes \varphi_{U}$. By \autoref{lem:error-characterization}, for a fixed $x,y\in [n]$ we have $$\left|\widetilde{\mathbf{C}}_{x,y}-(\widetilde{\mathbf{A}}\widetilde{\mathbf{B}})_{x,y}\right|^{2}=\left(\sum_{(i,k,k',j):\varphi(i,k,k',j)\in \mathrm{Col}^{\varphi_{S}(x),\varphi_{U}(y)}}\boldsymbol{\alpha}_{i}\boldsymbol{\beta}_{j}\boldsymbol{\gamma}_{k}\boldsymbol{\gamma}_{k'}\cdot \mathbf{A}_{i,k}\cdot \mathbf{B}_{k',j}\right)^{2}.$$
    
    For a quad-tuple $\mathbf{q}=(i,k,k',j)\in \mathcal{X}$ where $\mathcal{X}=\set{(i,k,k',j):k'\not=k}\subset [n]^{4}$, we define $f_{\mathbf{q}}$ to be the random variable indicating if $\varphi(\mathbf{q})\in \mathrm{Col}^{\varphi_{S}(x),\varphi_{U}(y)}$. Since the functions $\varphi_{S},\varphi_{T},\varphi_{U}$ are randomly and uniformly chosen, any assignment has equal probability, hence $\Pr(f_{\mathbf{q}}=1)\le \frac{\rho_{n}(G)-n^{3}}{n^{4}-n^{3}}$ (this is the probability of $\mathbf{q}$ falling into some collisions set, which upper bounds the probability of $\mathbf{q}$ falling into any specific collisions set).

    Furthermore, we write $g_{\mathbf{q}}=\boldsymbol{\alpha}_{i}\boldsymbol{\beta}_{j}\boldsymbol{\gamma}_{k}\boldsymbol{\gamma}_{k'}\cdot \mathbf{A}_{i,k}\cdot \mathbf{B}_{k',j}$ and note that $\mathbb{E}[g_{\mathbf{q}}\cdot g_{\mathbf{p}}]\not=0$ only when $\mathbf{p}=\mathbf{q}$ or the roles of $k,k'$ are switched, meaning $\mathbf{p}=(i,k',k,j)$. In the first case $\mathbb{E}[g_{\mathbf{q}}g_{\mathbf{p}}]=\mathbf{A}_{i,k}^{2}\cdot \mathbf{B}_{k',j}^{2}$. In the second case, observe that $f_{\mathbf{q}}$ and $f_{\mathbf{p}}$ are both $1$ iff $$\varphi_{S}(i)^{-1}\varphi_{T}(k)\varphi_{T}(k')^{-1}\varphi_{U}(j)=\varphi_{S}(x)^{-1}\varphi_{U}(y)=\varphi_{S}(i)^{-1}\varphi_{T}(k')\varphi_{T}(k)^{-1}\varphi_{U}(j)$$which implies $\varphi_{T}(k)\varphi_{T}(k')^{-1}$ has order two. By assumption, this can't happen, thus the indicators cannot be both one. In particular, $\mathbb{E}[f_{\mathbf{q}}f_{\mathbf{p}}]=0$. 
    
    We are thus left with 
    \begin{align*}
        \mathbb{E}[(\widetilde{\mathbf{C}}_{x,y}-(\widetilde{\mathbf{A}}\widetilde{\mathbf{B}})_{x,y})^{2}]&= \sum_{\mathbf{p},\mathbf{q}: \varphi(\mathbf{p}),\varphi(\mathbf{q})\in \mathrm{Col}^{\varphi_{S}(x),\varphi_{U}(y)}}\mathbb{E}[f_{\mathbf{p}}f_{\mathbf{q}}]\cdot \mathbb{E}[g_{\mathbf{p}}\cdot g_{\mathbf{q}}]\\
        &= \sum_{\mathbf{q}:\varphi(\mathbf{q})\in \mathrm{Col}^{\varphi_{S}(x),\varphi_{U}(y)}}\mathbb{E}[f_{\mathbf{q}}^{2}]\cdot \mathbf{A}_{i,k}^{2}\mathbf{B}_{k',j}^{2}\\
        &\le \left(\frac{\rho_{n}(G)-n^{3}}{n^{4}-n^{3}}\right)\sum_{\mathbf{q}:\varphi(\mathbf{q})\in \mathrm{Col}^{\varphi_{S}(x),\varphi_{U}(y)}} \mathbf{A}_{i,k}^{2}\mathbf{B}_{k',j}^{2}.
    \end{align*}
    
    Summing over all $x,y\in[n]$ and using the fact the collision sets are disjoint, we obtain 
    \begin{align*}
        \mathbb{E}\left\lVert \mathbf{C}-\mathbf{A}\mathbf{B} \right\rVert_{F}^{2}&= \mathbb{E}\left\lVert \widetilde{\mathbf{C}}-\widetilde{\mathbf{A}}\widetilde{\mathbf{B}} \right\rVert_{F}^{2}\\
        &= \sum_{x,y\in [n]}\sum_{\mathbf{q}\in \mathrm{Col}^{\varphi_{S}(x),\varphi_{U}(y)}}\frac{\rho_{n}(G)-n^{3}}{n^{4}-n^{3}}\cdot \mathbf{A}_{i,k}^{2}\mathbf{B}_{k',j}^{2}\\
        &= \frac{\rho_{n}(G)-n^{3}}{n^{4}-n^{3}}\cdot\sum_{\mathbf{q}\in \mathrm{Col}}\mathbf{A}_{i,k}^{2}\mathbf{B}_{k',j}^{2}\\
        &\le \frac{\rho_{n}(G)-n^{3}}{n^{4}-n^{3}}\cdot\sum_{\mathbf{q}\in [n]^{4}}\mathbf{A}_{i,k}^{2}\mathbf{B}_{k',j}^{2}\\
        &= \frac{\rho_{n}(G)-n^{3}}{n^{4}-n^{3}}\cdot \left\lVert \mathbf{A} \right\rVert_{F}^{2}\cdot \left\lVert \mathbf{B} \right\rVert_{F}^{2},
    \end{align*}
    as desired.
\end{proof}

\paragraph{Simple Instantiation using Univariate Polynomials.}
Let us now apply this theorem to a very simple construction in the group $\Z_{rn^2}$. Define $$S=\set{rn\cdot i\mid i\in [n]}, \quad T=\set{(r-1)\cdot i\mid i\in [n]},\quad U=\set{r\cdot i\mid i\in [n]}.$$
A simple calculation shows that in this case we have $\rho_{n,\Z_{rn^2}}(S,T,U)\le \frac{2n^4}{r}$, using the fact $S+U$ is the arithmetic progression with difference $r$ and $T-T$ has a very simple structure (a subset of the arithmetic progression with difference $r-1$, with multiplicities determined by a triangular distribution).

\begin{corollary}
    Running PolyForm in $\Z_{rn^2}$ is optimal for abelian groups (up to constant factors). In particular, we have an algorithm achieving expected error $O(\frac{1}{r})$ with running time $O(rn^2\log n)$.
\end{corollary}

\subsection{Extension to Approximate STPP}
Building on the intuition that STPP constructions are better at packing information (they are non-trivial in abelian groups, while TPP constructions are trivial), we conjecture that generalizing the approximate definition to the STPP case might lead to better tradeoffs.

However, in the case of STPP, one has to account for the mixing of the different matrix supports, which complicates things. We suggest the following natural generalization of \nameref{def:atpq} for the simultaneous case.

\begin{definition}[SIF]\label{def:sif}
    Let $G$ be a finite group. A simultaneous indexing family (SIF) of length $k$ in $G$ is a sequence $\mathcal{I}=\set{(S_{i},T_{i},U_{i})}_{i=1 }^{k}$, satisfying:
    \begin{enumerate}
        \item For each $i$, the triplet $(S_{i},T_{i},U_{i})$ is an \nameref{def:indexing-triplet} in $G$, meaning that $\left|S_{i}^{-1}T_{i}\right|=\left|S_{i}\right|\cdot \left|T_{i}\right|$ and so on.
        \item For any $i\not=j\in [k]$, the sets are disjoint: $$S_{i}^{-1}T_{i}\cap S_{j}^{-1}T_{j}=\emptyset,\quad T_{i}^{-1}U_{i}\cap T_{j}^{-1}U_{j}=\emptyset,\quad S_{i}^{-1}U_{i}\cap S_{j}^{-1}U_{j}=\emptyset.$$
    \end{enumerate}

    We define the \textbf{size} of $\cI$, denoted $\mu(\cI)$, to be $\sum_{i=1}^{k}\left|S_{i}\right|\left|T_{i}\right|\left|U_{i}\right|$. We define the \textbf{length} of $\cI$, denoted $\ell(\cI)=k$ to be $k$. Moreover, we define the \textbf{input capacity} of $\cI$, denoted $\kappa(\cI)$, to be $\left(\sum_{i}\left|S_{i}\right|\left|T_{i}\right|\right)\cdot \left(\sum_{i}\left|T_{i}\right|\left|U_{i}\right|\right)$.
\end{definition}

This ensures that we can embed the $k$ simultaneous matrices into the group and read off the results.

\begin{definition}[ASTPQ]\label{def:astpq}
    Let $G$ be a finite group and $\mathcal{I}$ an \nameref{def:sif} of length $k$. Denote $M=\mu(\mathcal{I})$. The \textit{Approximate Simultaneous Triple Product Quantity} (ASTPQ) for $\mathcal{I}$, denoted by $\widetilde{\rho}_{M,G}(\mathcal{I})$ is defined to be the number of solutions to the equation $$s_{1}s_{2}^{-1}t_{1}t_{2}^{-1}u_{1}u_{2}^{-1}=1,$$where the solutions are tuples $(s_{1},s_{2},t_{1},t_{2},u_{1},u_{2})$, with $i,j,\ell\in [k]$ and $$s_{1}\in S_{\ell},s_{2}\in S_{i}, t_{1}\in T_{i},t_{2}\in T_{j}, u_{1}\in U_{j}, u_{2}\in U_{\ell}.$$
    The ASTPQ for $G$ and size $M$ is defined to be $\widetilde{\rho}_{M}(G)=\min_{\mathcal{I},\mu(\mathcal{I})=M}=\widetilde{\rho}_{M,G}(\mathcal{I})$. In other words, we take the minimal ASTPQ over all SIFs of size $M$. Note that the length is not constant, a priori.
\end{definition}

\begin{remark}
    The number of trivial solutions, where $i=j=\ell$ and $s_{1}=s_{2},t_{1}=t_{2}$ and $u_{1}=u_{2}$, is exactly $M$. If $\widetilde{\rho}_{M,G}(\mathcal{I})=M$ then $\mathcal{I}$ is actually an STPP configuration for $G$.
\end{remark}

In this setting, our algorithm outputs $\ell(\cI)$ matrices $\set{\bC_i}_{i=1}^{\ell(\cI)}$, and the error of computing a convolution is $\sum_{i=1}^{\ell(\cI)}\norm{\bC_i-\bA_i\bB_i}{F}^2$, which we normalize against $(\sum_i \norm{\bA_i}{F}^2)\cdot (\sum_i\norm{\bB}{F}^2)$. Following a similar proof to that of \autoref{thm:polyform-error-lower-bound} (almost verbatim), we can derive the following lower bound: 
\begin{theorem}
    Let $G$ be a finite group and $\cI$ an \nameref{def:sif} of size $\mu(\cI)=M$ and capacity $\kappa(\cI)$ Then the expected randomized error of implementing the algorithm in $G$ is $\Omega\left(\frac{\widetilde{\rho}_{M,G}(\mathcal{I})-M}{\kappa(\mathcal{I})}\right)$.
\end{theorem}

The size of $\tilde{\rho}$ becomes much harder to analyze when $k>1$ (note $k=1$ is the \nameref{def:atpq} case), due to the interference of different sets with each other. A very simple lower bound, is obtained by ignoring inter-set collisions. By \autoref{cor:diffset_LB}, we obtain $\tilde{\rho}_{M,G}(\cI)\ge \para{\sum_{i=1}^k\abs{S_i}^2\abs{T_i}^2\abs{U_i}^2}/\abs{G}$. In the case $\cI$ is just a sum of square matrix multiplications of shared size $m$, we have $\tilde{\rho}_{km^3,G}(\cI)\ge \frac{km^6}{\abs{G}}$.

In general, a more refined lower bound on $\tilde{\rho}$ is harder to obtain. We argue that the following Fourier analytic approach highlights why this problem is non-trivial. Write $\hat{X}$ to be the Fourier transform of the indicator of a set $X\subset G$. A simple application of the Fourier inversion formula leads to the identity $$\tilde{\rho}_{M,G}(\cI)=\frac{1}{\abs{G}}\sum_{\chi\in \hat{G}}\para{\sum_{i=1}^k\hat{S_i}(\chi)\cdot \hat{T}_i (\chi)}\para{\sum_{i=1}^k\hat{T_i}(\chi)\cdot \hat{U}_i (\chi)}\para{\sum_{i=1}^k\hat{S_i}(\chi)\cdot \hat{U}_i (\chi)}.$$
The term corresponding to the trivial character is $$\frac{1}{\abs{G}}\para{\sum_{i=1}^k\abs{S_i}\cdot \abs{T_i }}\para{\sum_{i=1}^k\abs{T_i}\cdot \abs{U_i }}\para{\sum_{i=1}^k\abs{S_i}\cdot \abs{U_i }},$$which is $\frac{kM^2}{\abs{G}}=\frac{k^3m^6}{\abs{G}}$ when $\cI$ is $k$ square matrix multiplications of size $m$ (a factor of $k^2$ over the trivial lower bound). The question now becomes, how large is the rest of the sum (the non-trivial characters), and in particular, when isn't it negative?

This question was studied in additive combinatorics, and relates to a property called \textit{Quasi-Randomness} of sets (see \cite{Bre14} and references within for background). In general, quasi-random sets have rapid mixing properties, which can be proved by showing the sum of non-trivial characters in the above decomposition is \textit{small} in magnitude. Thus, as a rule of thumb, quasi-randomness implies higher $\tilde{\rho}$. Since our goal is to minimize $\tilde{\rho}$, this analysis suggests we are looking for highly structured and correlated sets, which are as far from quasi-random sets as possible (by making the sum over non-trivial characters as negative as possible). As such, finding good candidate sets is a hard and interesting problem.

\section{Towards Group-Theoretic Practical AMM}

Perhaps the most important and interesting question  
that arises from this work, is whether the STPP construction used in \autoref{thm_S&S_AMM} and \autoref{thm_main_exactMM_informal} can be made to work for real-life sized matrices (say $n=8K$, as in LLMs \cite{STL25}). Although FFT is a practical algorithm, the asymptotics of the algorithms are bad. One sub-optimal feature that we point out in \autoref{thm_S&S_AMM} is that it uses the STPP construction in a \textit{black-box} manner. This inevitably requires the compressed matrices to match the asymptotics of the exact construction in \autoref{thm_main_exactMM_informal}. Schematically, this can be viewed as:
\begin{align}\label{sns-schematic}
\mathrm{S\&S}(\bA,\bB)=\mathrm{STPP}(\mathrm{Sketch}(\bA),\mathrm{Sketch}(\bB)).
\end{align}
Note that \autoref{sns-schematic} can be viewed as a special case of a low-degree approximation of the STPP polynomials (in this case, using STPP polynomials of lower degree, with linearly sketched coefficients). Another approach for degree reduction, is to instead sketch the \textbf{convolution} of the polynomials (e.g., via \cite{Nakos20}), i.e.,
\begin{align}
    \mathrm{SketchSTPP}(\bA,\bB)=\bW^{\top}\Psi(\Phi \bU a \odot \Phi \bV b).
\end{align}
Here $\bU,\bV,\bW$ are the bilinear representation for the embed--convolve--decode algorithm, and $\Phi,\Psi$ are sketching matrices. The upshot of this approach is that it circumvents the asymptotics of the STPP based algorithm (\autoref{thm_main_exactMM_informal}), assuming $\Phi\cdot  \bU \cdot a$ (and $\Phi\cdot \bV\cdot  b$) can be computed \textit{implicitly} in $O(rn^2)$.
One special case, is Fourier domain sparsification, where $\Phi$ is a projection onto a subset of \textit{frequencies}. Unfortunately, while this sketch can be computed efficiently (in time independent of the number of rows of $\bV$), it does not yield better error than the standard error \eqref{eq_AMM_sota} for general $\bA,\bB$.

The definition of the \nameref{def:atpq} follows the same idea, and can be viewed as applying an approximate convolution for the polynomials through a dimension reduction and computation of the exact convolution in a smaller dimension. The main open question that remains in this approach, is studying the generalization to the STPP case, and in particular, finding constructions.

\ifnotanonymous
\subsection*{Acknowledgments}
We thank Kevin Pratt for insightful discussions and for pointing out the sub-cubic algorithm for exact matrix multiplication.
We also thank Zachary Chase  for insightful discussions on the PolyForm scheme and for originating the proof of the lower bound (\autoref{cor:diffset_LB}).
\fi

\printbibliography

\newpage
\appendix
\section{Fast Combinatorial Algorithm}\label{app:fast-cmm-algorithm}
For convenience we introduce the following notation - for $i\in\left\{ 1,2,3\right\} $ and $\ba\in\left\{ 0,1\right\} ^{N}$ we let $S_{i}^{\ba}$ denote the (cartesian) product $S_{i}^{\ba}:=\prod_{j=1}^{N}S_{i}^{\ba_j}$. We also let $\bin\left(x\right)\in\left\{ 0,1\right\} ^{N}$ denote the binary representation of a number $x\in\left[2^{N}\right]$ (for convenience we identify $0$ with $2^{N}$).

\begin{algorithm}[ht]
\caption{$\Z_{m}^{3N}$ STPP algorithm ($m,N$)}\label{alg:zm3N-algorithm}
\begin{algorithmic}[1]
\Require Matrices $\mathbf{A}_{1},\ldots,\mathbf{A}_{2^{N}},\mathbf{B}_{1},\ldots,\mathbf{B}_{2^{N}}$
of size $\left(m-1\right)^{N}\times\left(m-1\right)^{N}$.
\Ensure Product matrices $\mathbf{A}_{1}\mathbf{B}_{1},\ldots ,\bA_{2^N}\bB_{2^N}$.
\For{$x=1,\ldots ,2^N$}
\State{Index $\mathbf{A}_{i}$ by $S_{2}^{\bin(x)}-S_{1}^{\bin(x)}$}\Comment{Difference Set}
\EndFor
\State{Compute the FFT of the resulting signal $a$ and denote the result by $\hat{a}$}
\For{$x=1,\ldots,2^N$}
\State{Index $\mathbf{B}_{i}$ by $S_3^{\bin(x)}-S_2^{\bin(x)}$}
\EndFor
\State{Compute the FFT of the resulting signal $b$ and denote the result by $\hat{b}$}
\State{Compute the element-wise product $\hat{c}=\hat{a}\odot \hat{b}$}
\State{Compute the inverse FFT of $\hat{c}$, denoted by $c$}
\For{$x=1,\ldots,2^N$}
\State{Read $\bA_{x}\bB_x$ from the coefficients of $S_3^{\bin(x)}-S_1^{\bin(x)}$}
\EndFor\\
\Return{$\mathbf{A}_{1}\bB_1,\ldots,\bA_{2^n}\bB_{2^N}$}
\end{algorithmic}
\end{algorithm}

\section{Omitted Proofs}\label{app:proofs}

An important and useful property is:

\begin{lemma}[\cite{CU03}]\label{lem:tpp-product-property}
    If $\ip{S,T,U}$ satisfy the TPP in $G$ and $\ip{S',T',U'}$ in $G'$, then $\ip{S\times S',T\times T',U\times U'}$ satisfies the TPP in the product group $G\times G'$. Similarly, if $\set{\ip{S_i,T_i,U_i}}_{i=1}^k$ satisfy the STPP in $G$ and $\set{\ip{S_i',T_i',U_i'}}_{i=1}^{k'}$ in $G'$, then $\set{\ip{S_i\times S_j',T_i\times T_j', U_i\times U_j'}}_{i\in [k],j\in [k']}$ satisfy STPP in $G\times G'$.
\end{lemma}

\begin{lemma}
    Let $d_{\min} = \min\set{d_{\pi}\mid \pi\in \mathrm{Irr}(G), d_{\pi}>1}$. Then for an indexing triplet $S,T,U\subset G$ we have $$\rho_{n,G}(S,T,U)\ge \frac{(\left|S\right|\left|T\right|\left|U\right|)^{2}}{\left|G\right|}-\left|S\right|\left|T\right|\left|U\right|\cdot \sqrt{\frac{\left|G\right|}{d_{\min}}}.$$
\end{lemma}
\begin{proof}
    The proof follows \cite{Bre14} Lemma 2.2.
    
    Write $f=1_{S^{-1}T}*1_{T^{-1}U}*1_{S^{-1}U}$ and note that $f(e_G)=\rho_{n,G}(S,T,U)$. By Fourier inversion, we have 
    \begin{align}\label{lem-eq-1}
    f(e_G)=\frac{1}{\abs{G}}\sum_{\pi}d_{\pi} \ip{\pi(f),\pi(e_G)}=\frac{1}{\abs{G}}\sum_{\pi}d_{\pi}\mathrm{Tr}({\pi(S^{-1})\pi(T)\pi(T^{-1})\pi(U)\pi(U^{-1})\pi(S))}
    \end{align}
    where the sum is over all irreducible representations $\pi$ of $G$ with dimension $d_{\pi}$ ($\pi:G\to \mathrm{GL}(d_\pi,\C)$), $\pi(f)=\sum_{g\in G}f(g)\pi(g)$ and $\pi(X)=\sum_{x\in X}\pi(x)$ for $X\subset G$. Here the inner product of matrices is just $\ip{X,Y}=\mathrm{Tr}(XY^*)$ (the induced norm is just the Frobenius norm, i.e., $\sqrt{\ip{X,X}}=\norm{X}{F}$).
    By Parseval's identity applied to $1_{S^{-1}T}$ we obtain that 
    \begin{align}\label{lem-eq-2}
    \abs{S}\cdot \abs{T}=\abs{S^{-1}T}=\sum_{g\in G}\abs{1_{S^{-1}T}(g)}^2=\frac{1}{\abs{G}}\sum_{\pi}d_{\pi}\cdot \norm{\pi(1_{S^{-1}T})}{}^2.
    \end{align}
    Therefore, for every $\pi$ with $d_{\pi}>1$ we have $\norm{\pi(1_{S^{-1}T})}{}^2\le \frac{\abs{S^{-1}T}\cdot \abs{G}}{d_{\min}}$.
    When $d_{\pi}$ we see that $\pi(S)\pi(S^{-1})=\abs{\pi(S)}^2$ is just a non-negative number. Moreover, if $\pi$ is the trivial representation, this is just $\abs{S}^2$.
    
    Thus, if we split the sum in \autoref{lem-eq-1} to a main term, which is the trivial representation, and a remainder term which is a sum over $d_{\pi}>1$, we obtain $$f(e_G)\ge \frac{\abs{S}^2\cdot \abs{T}^2\cdot \abs{U}^2}{\abs{G}}+\frac{1}{\abs{G}}\sum_{d_{\pi}>1}d_{\pi} \mathrm{Tr}(\pi(S^{-1})\pi(T)\pi(T^{-1})\pi(U)\pi(U^{-1})\pi(S))$$and by applying the Cauchy-Schwarz inequality we obtain $\mathrm{Tr}(\pi(S^{-1})\pi(T)\pi(T^{-1})\pi(U)\pi(U^{-1})\pi(S))\ge -\norm{\pi(S^{-1}T)}{}\cdot \norm{\pi(T^{-1}U)}{}\cdot \norm{\pi(U^{-1}S)}{}$, and together with \autoref{lem-eq-2} we get 
    \begin{align*}
    f(e_G)& \ge \frac{(\abs{S}\abs{T}\abs{U})^2}{\abs{G}}-\frac{1}{\abs{G}}\sum_{d_{\pi>1}}d_{\pi}\norm{\pi(S^{-1}T)}{}\cdot \norm{\pi(T^{-1}U)}{}\cdot \norm{\pi(U^{-1}S)}{} \\
    & \ge \frac{(\abs{S}\abs{T}\abs{U})^2}{\abs{G}}-\frac{1}{\abs{G}}\sum_{d_{\pi>1}}d_{\pi}\sqrt{\frac{\abs{S}\abs{T}\abs{G}}{d_{\min}}}\cdot \norm{\pi(T^{-1}U)}{}\cdot \norm{\pi(U^{-1}S)}{} \\
    & \ge \frac{(\abs{S}\abs{T}\abs{U})^2}{\abs{G}}-\frac{1}{\abs{G}} \cdot \sqrt{\frac{\abs{S}\abs{T}\abs{G}}{d_{\min}}} \sqrt{\sum_{d_{\pi>1}}d_{\pi}\cdot \norm{\pi(T^{-1}U)}{}^2}\cdot \sqrt{\sum_{d_{\pi}>1} d_{\pi}\cdot  \norm{\pi(U^{-1}S)}{}^2} \\
    & = \frac{(\abs{S}\abs{T}\abs{U})^2}{\abs{G}}-\frac{1}{\abs{G}} \cdot \sqrt{\frac{\abs{S}\abs{T}\abs{G}}{d_{\min}}} \sqrt{\abs{T} \abs{U} \abs{G}}\cdot \sqrt{\abs{S}\abs{U}\abs{G}}\\
    & = \frac{(\left|S\right|\left|T\right|\left|U\right|)^{2}}{\left|G\right|}-\left|S\right|\left|T\right|\left|U\right|\cdot \sqrt{\frac{\left|G\right|}{d_{\min}}}
    \end{align*}
\end{proof}

\end{document}